\documentclass[11pt]{article}
\usepackage{amssymb}
\usepackage{amsmath}
\usepackage{amsthm}
\usepackage{graphicx}

\usepackage{cite}

\theoremstyle{plain}
\newtheorem{theorem}{Theorem}

\newtheorem{lemma}[theorem]{Lemma}
\newtheorem{corollary}[theorem]{Corollary}

\theoremstyle{definition}
\newtheorem{definition}{Definition}
\numberwithin{theorem}{section}
\numberwithin{definition}{section}
\newtheorem{problem}{Open Problem}

\DeclareMathOperator{\rank}{rk}
\DeclareMathOperator{\Prob}{Pr}

\DeclareMathOperator{\row}{Row}

\DeclareMathOperator{\GL}{GL}

\newcommand{\bF}{\mathbb{F}}
\newcommand{\bs}{\boldsymbol}
\newcommand{\qbinom}[2]{\genfrac{[}{]}{0pt}{}{\,#1\,}{\,#2\,}_{\!q}} 

\begin{document}

\title{Finite Field Matrix Channels for Network Coding}
\author{Simon R. Blackburn and Jessica Claridge\\
Department of Mathematics \\
Royal Holloway University of London\\
Egham, Surrey TW20 0EX, United Kingdom}
\date{}
\maketitle


\abstract
In 2010, Silva, Kschischang and K\"otter studied certain classes of finite field matrix channels in order to model random linear network coding where exactly $t$ random errors are introduced.

In this paper we consider a generalisation of these matrix channels where the number of errors is not required to be constant, indeed the number of errors may follow any distribution. We show that a capacity-achieving input distribution can always be taken to have a very restricted form (the distribution should be uniform given the rank of the input matrix). This result complements, and is inspired by, a paper of Nobrega, Silva and Uchoa-Filho, that establishes a similar result for a class of matrix channels that model network coding with link erasures. Our result shows that the capacity of our channels can be expressed as a maximisation over probability distributions on the set of possible ranks of input matrices: a set of linear rather than exponential size.
\endabstract


\section{Introduction} \label{intro}

Network coding, first defined in \cite{ahlswede2000}, allows intermediate nodes of a network to compute with and modify data, as opposed to the traditional view of nodes as `on/off' switches. This can increase the rate of information flow through a network. It is shown in \cite{li2003linear} that linear network coding is sufficient to maximise information flow in multicast problems, that is when there is one source node and information is to be transmitted to a set of sink nodes. Moreover, in \cite{ho2006RLNC} it is shown that for general multisource multicast problems, random linear network coding achieves capacity with probability exponentially approaching $1$ with the code length. 

In random linear network coding, the source injects packets into the network; these packets can be thought of as vectors of length $m$ with entries in a finite field~$\bF_q$ (where $q$ is a fixed power of a prime). The packets flow through a network of unknown topology to a sink node. Each intermediate node forwards packets that are random linear combinations of the packets it has received. A sink node attempts to reconstruct the message from these packets. In this context, Silva, Kschischang and K\"otter~\cite{silva2010} studied a channel defined as follows. We write $\mathbb{F}_q^{n \times m}$ to denote the set of all $n \times m$ matrices over $\bF_q$, and write $\GL(n,q)$ for the set of all invertible matrices in $\mathbb{F}_q^{n \times n}$.

\begin{definition}
The \emph{Multiplicative Matrix Channel} (MMC) has input set $\mathcal{X}$ and output set $\mathcal{Y}$, where  $\mathcal{X}=\mathcal{Y}=\mathbb{F}_q^{n \times m}$. The channel law is
\[
\boldsymbol{Y} =\boldsymbol{A}\boldsymbol{X}
\]
where $\boldsymbol{A} \in\GL(n,q)$ is chosen uniformly at random.
\end{definition} 
Here the rows of $\boldsymbol{X}$ correspond to the packets transmitted by the source, the rows of $\boldsymbol{Y}$ are the packets received by the sink, and the matrix $\boldsymbol{A}$ corresponds to the linear combinations of packets computed by the intermediate nodes.

Inspired by Montanari and Urbanke~\cite{montanari2013}, Silva \emph{et al} modelled the introduction of random errors into the network by considering the following generalisation of the MMC. We write $\mathbb{F}_q^{n \times m,r}$ for the set of all $n \times m$ matrices of rank $r$.
\begin{definition}
The \emph{Additive Multiplicative Matrix Channel with $t$ errors} (AMMC) has  input set $\mathcal{X}$ and output set $\mathcal{Y}$, where  $\mathcal{X}=\mathcal{Y}=\mathbb{F}_q^{n \times m}$. The channel law is
\[
\boldsymbol{Y}=\boldsymbol{A}(\boldsymbol{X}+\boldsymbol{B})
\]
where $\boldsymbol{A}\in\GL(n,q)$ and $\boldsymbol{B}\in \mathbb{F}_q^{n \times m,t}$ are chosen uniformly and independently at random.
\end{definition}
So the matrix $\boldsymbol{B}$ corresponds to the assumption that exactly~$t$ linearly independent random errors have been introduced. The MMC is exactly the AMMC with zero errors.

We note that the AMMC is very different from the error model studied in the well-known paper by K\"otter and Kschischang~\cite{Koetter2008}, where the errors are assumed to be adversarial (so the worst case is studied).

In \cite{silva2010} the authors give upper and lower bounds on the capacity of the AMMC, which are shown to converge in certain interesting limiting cases. The exact capacity of the AMMC for fixed parameter choices is hard to determine due to the many degrees of freedom involved: the naive formula maximises over a probability distribution on the set of possible input matrices, and this set is exponentially large.

In this paper we consider a generalisation of these matrix channels that allows the modelling of channels were the number of errors is not necessarily fixed. (For example, it enables the modelling of situations when at most $t$ errors are introduced, or when the errors are not necessarily linearly independent, or both.) To define our generalisation, we need the following notation which is due to Nobrega, Silva and Uchoa-Filho~\cite{Nobrega2013}.
\begin{definition}
Let $\mathcal{R}$ be a probability distribution on the set $\{0,1,\ldots,\min\{m,n\}\}$ of possible ranks of matrices $M\in\mathbb{F}_q^{n \times m}$. We define a distribution on the set $\mathbb{F}_q^{n \times m}$ of matrices by choosing $r$ according to $\mathcal{R}$, and then once $r$ is fixed choosing a matrix $M\in \mathbb{F}_q^{n \times m,r}$ uniformly at random. We say that this distribution is \emph{Uniform Given Rank (UGR) with rank distribution $\mathcal{R}$}. We say a 
distribution on $\mathbb{F}_q^{n \times m}$ is \emph{Uniform Given Rank (UGR)} if it is UGR with rank distribution $\mathcal{R}$ for some distribution $\mathcal{R}$.
\end{definition}
We write $\mathcal{R}(r)$ for the probability of rank $r$ under the distribution $\mathcal{R}$. So a distribution on $\mathbb{F}_q^{n \times m}$ is UGR with rank distribution $\mathcal{R}$ if and only if each $M\in \mathbb{F}_q^{n \times m}$ of rank $r$ is chosen with probability $\mathcal{R}(r)/|\mathbb{F}_q^{n \times m,r}|$.
\begin{definition}
Let $\mathcal{R}$ be a probability distribution on the set $\{0,1,\ldots,\min\{m,n\}\}$ of possible ranks of matrices $M\in\mathbb{F}_q^{n \times m}$. The \emph{Generalised Additive Multiplicative MAtrix Channel with rank error distribution $\mathcal{R}$ (the Gamma channel $\Gamma(\mathcal{R})$)} has  input set $\mathcal{X}$ and output set $\mathcal{Y}$, where  $\mathcal{X}=\mathcal{Y}=\mathbb{F}_q^{n \times m}$. The channel law is
\[
\boldsymbol{Y}=\boldsymbol{A}(\boldsymbol{X}+\boldsymbol{B})
\]
where $\boldsymbol{A}\in\GL(n,q)$ is chosen uniformly, where $\boldsymbol{B}\in \mathbb{F}_q^{n \times m}$ is UGR with rank distribution $\mathcal{R}$, and where  $\boldsymbol{A}$ and $\boldsymbol{B}$ are chosen independently.
\end{definition}
We see that the AMMC is the special case of $\Gamma(\mathcal{R})$ when $\mathcal{R}$ is the distribution choosing rank $t$ with probability $1$. In~\cite[\S VI-D]{silva2010} the authors consider a generalisation of the AMMC which is exactly the Gamma channel in the case when the error rank is bounded by $t$, that is they consider $\Gamma(\mathcal{R})$ when $\mathcal{R}$ is any distribution with $\mathcal{R}(r)=0$ for all $r>t$. For $t=n$ this channel is identical to the Gamma channel. However, the authors consider $t$ as a maximum value for the error rank (thus will be taken to be less than $n$), whereas we consider any full rank distribution, with high rank having low probability. 

Our model covers several very natural situations (some of which are also covered by the generalised AMMC with $t<n$). For example, we may drop the assumption that the $t$ errors are linearly independent by defining $\mathcal{R}(r)$ to be the probability that $t$ vectors span a subspace of dimension $r$, when the vectors are chosen uniformly and independently. We can also extend this to model situations when the number $t$ of vectors varies according to natural distributions such as the binomial distribution (which arises when, for example, a packet is corrupted with some fixed non-zero probability). In practice, given a particular network, one may run tests on the network to see the actual error patterns produced and define an empirical distribution on ranks. One could also define an appropriate distribution by considering some combination of the situations described.

We are interested in the capacity of the Gamma channel. In the generalised AMMC~\cite[\S VI.D]{silva2010} the authors establish a lower bound on the capacity that is at most $\log_q(t+1)$ lower than the capacity of the AMMC with the same value of $t$. Therefore in the limiting cases considered their generalised channel performs at least as well as the AMMC. This is a very useful result when $t$ is significantly smaller than $n$. However, if we take $t=n$ then the expressions~\cite[Eq. 19 \& 20]{silva2010} for the capacity in the limiting cases considered evaluate to zero.

Throughout this paper, we assume that $q$, $n$, $m$ and $\mathcal{R}$ are fixed by the application. We will refer to these values as the \emph{channel parameters}.

We note that the Gamma channel assumes that the \emph{transfer matrix} $\boldsymbol{A}$ is always invertible. This is a realistic assumption in random linear network coding in standard situations: the field size $q$ is normally large, which means linear dependencies between random vectors are less likely. 

In both~\cite{Nobrega2013} and~\cite{siavoshani2011} the authors consider (different) generalisations of the MMC channel that do not necessarily have a square full rank transfer matrix. Such channels allow modelling of network coding when no erroneous packets are injected into the network, but there may be link erasures. 
In~\cite{Nobrega2013}, Nobrega, Silva and Uchoa-Filho define the transfer matrix to be picked from a UGR distribution. One result of~\cite{Nobrega2013} is that a capacity-achieving input distribution for their class of MMC channels can always be taken to be UGR. 

A main result of this paper (Theorem~\ref{capacityUGR}) is that a capacity-achieving input distribution for Gamma channels can always be taken to be UGR. Theorem~\ref{capacityUGR} is a significant extension of the result of~\cite{Nobrega2013} to a new class of channels; the extension requires new technical ideas. This result is in contrast to the coding schemes proposed in \cite{silva2010}, which restrict input matrices to have a specific form, and achieves capacity in the limiting cases considered. Their restrictions on the input allows for straightforward decoding, whereas it is not immediately obvious of how to construct an efficient UGR coding scheme which achieves capacity for any given parameters, indeed this is a problem of interest for future research. 

Corollary~\ref{capacity_corollary} to the main result of the paper provides an explicit method for computing the capacity of Gamma channels which maximises over a probability distribution over the set of possible ranks of input matrices, rather than the set of all input matrices itself. Thus we have reduced the problem of computing the capacity of a Gamma channel to an explicit maximisation over a set of variables of linear rather than exponential size. As examples of the results of this approach, the table below gives the computed capacity $C$ of the AMMC channel with $2$ errors for $n\times 2n$ matrices over $\mathbb{F}_2$; the capacity $C'$ of the Gamma channel  for $n\times 2n$ matrices over $\mathbb{F}_2$ when the number of errors is binomially distributed with expected number of errors equal to $2$; and the capacity $C''$ when the number of errors is  $0$, $1$ or $2$ with probabilities $1/7$, $2/7$ and $4/7$ respectively. 
\[
\begin{array}{c|ccccccccc}
n&3&4&5&6&7&8&9&10\\\hline
C&1.731&5.586&11.644&18.807&30.050&42.381&56.798&73.290\\
C'&1.117&4.633&10.422&18.368&28.395&40.491&54.676&70.996\\
C''&2.765&7.354&14.090&23.003&34.032&47.094&62.177&79.274
\end{array}
\]
Figure~\ref{fig:AMMC_capacity} plots the capacity $C$ of the AMMC channel together with general upper and lower bounds on the capacity. (These bounds are due to Silva et al.~\cite[Theorem~6 and~7]{silva2010}. We comment that an improved upper bound due to Claridge~\cite[Equation~6.6.2]{Claridge_thesis} is very similar to~\cite[Theorem~6]{silva2010} for these parameters.) Similarly, Figure~\ref{fig:twoerror_capacity} plots the capacity $C''$ of the third example channel together with the lower bound on the capacity due to Silva et al.~\cite[\S~VI.D]{silva2010}.
\begin{figure}
\begin{center}
\includegraphics[width=8cm]{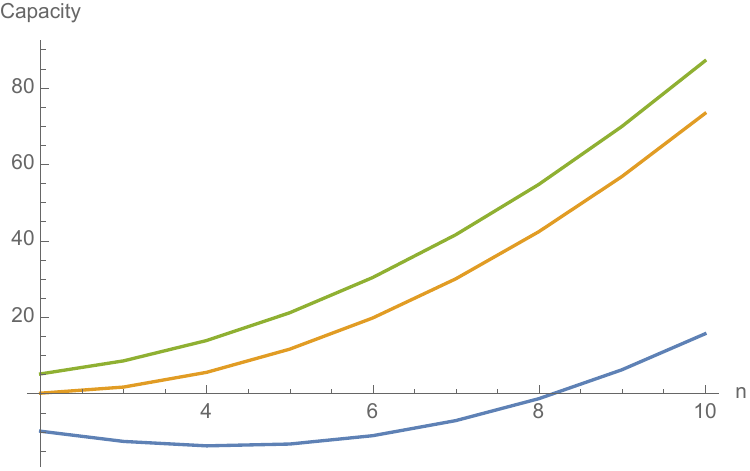}
\end{center}
\caption{Capacity (in bits) of the AMMC channel with $2$ errors for $n\times 2n$ matrices over $\mathbb{F}_2$. The three curves are: the upper bound from~\cite[Theorem~6]{silva2010}; the capacity computed using methods in this paper; the lower bound from~\cite[Theorem~7]{silva2010}.}
\label{fig:AMMC_capacity}
\end{figure}
\begin{figure}
\begin{center}
\includegraphics[width=8cm]{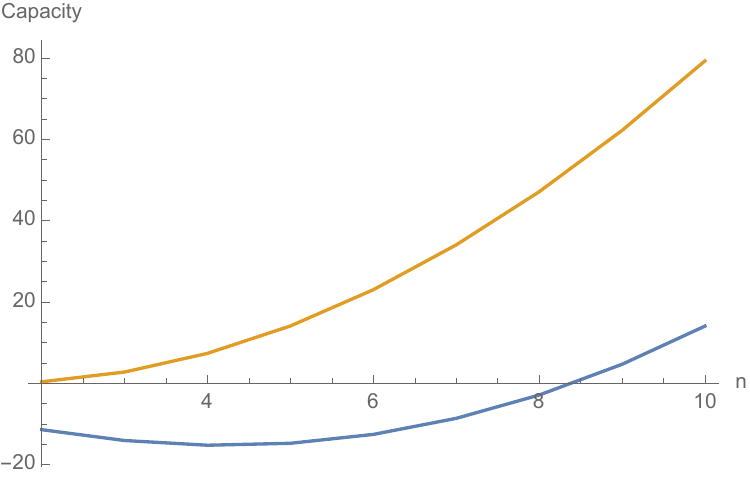}
\end{center}
\caption{Capacity $C''$ (in bits) of the matrix channel with $0$, $1$ and $2$ errors for $n\times 2n$ matrices over $\mathbb{F}_2$ specified in the text, together with the lower bound on the capacity from~\cite[\S~VI.D]{silva2010}.}
\label{fig:twoerror_capacity}
\end{figure}

The remainder of the paper is organised as follows. Section~\ref{sec:prelims} proves some preliminary results needed in what follows. In Section~\ref{sec:matrices} we state results from matrix theory that we use. Section \ref{sec:Gamma_rank_dist} establishes a relationship between the distributions of the ranks of input and output matrices for a Gamma channel. Section~\ref{sec:Gamma_UGR} proves Theorem~\ref{capacityUGR}, and Section~\ref{sec:Gamma_capacity} proves Corollary~\ref{capacity_corollary}, giving an exact expression for the capacity of the Gamma channels. In Section~\ref{sec:matrices_proofs} we prove the results from matrix theory that we use in earlier sections. Finally, Section~\ref{sec:conclusion} contains some concluding~remarks.


\section{Preliminaries on finite-dimensional vector spaces}
\label{sec:prelims}

In this section we discuss finite-dimensional vector spaces and consider several counting problems involving subspaces and quotient spaces. 

 The \textit{Gaussian binomial coefficient}, denoted $\qbinom{m}{d}$, is defined to be the number of $d$-dimensional subspaces of an $m$-dimensional space over $\mathbb{F}_q$. It is given by (e.g. \cite[\S 9.2]{Cameron1994})
\begin{equation}
\label{eq.def_GBCoeff}
\setlength{\nulldelimiterspace}{0pt}
\qbinom{m}{d}=\left\{%
\begin{array}{cc}
\displaystyle\prod_{i=0}^{d-1}\frac{(q^m-q^i)}{(q^d-q^i)},&\mathrm{for}\;\; d\leq m\\
0,&\mathrm{for}\;\; d>m.%
\end{array}\right.
\end{equation}

Let $V_1$ be a subspace of $V$. The following lemma gives the number of subspaces $U$ of $V$ where the intersection of $U$ and $V_1$, and the image of $U$ in the quotient space $V/V_1$, are both fixed.

\begin{lemma} \label{lemma_count_fix_int_image}
Let $V$ be a $d_V$-dimensional vector space. Let $V_1$, $V_2$ be subspaces of $V$, of dimensions $d_{V_1}$ and $d_{V_2}$ respectively, such that $V_2 \subseteq V_1$. The number of $d_U$-dimensional subspaces $U \subseteq V$ such that $U \cap V_1 =V_2$ and the image of $U$ in the quotient space $V/V_1$ is the fixed $d_U-d_{V_2}$ dimensional space $U'$, is given by
\[
q^{(d_{U}-d_{V_2})(d_{V_1} - d_{V_2})}.
\]
\end{lemma}

\begin{proof}
Fix a basis for $V_2$, say $\{ b_{1,1}, \dots , b_{1,d_{V_2}} \}$. 
Let $\pi : V \rightarrow V/V_1$ be the map which takes vectors in $V$ to their image in $V/V_1$. 
For $d_{U'}=d_{U}-d_{V_2}$, let $\{ y_1, \dots y_{d_{U'}} \}$ be a basis for $U'$, and let $\{b_{2,1}, \dots, b_{2,d_{U'}}\}$ be some vectors in $V$ such that $\pi(b_{2,i}) = y_i$, for $i=1, \dots , d_{U'}$. 

It is easy to check that every subspace $U$ of the form we are counting has a basis
\[
B=\{ b_{1,1}, \dots , b_{1,d_{V_2}} , v_1+b_{2,1}, \dots, v_{d_{U'}}+b_{2,d_{U'}} \}
\]
where $v_1,\dots, v_{d_{U'}} \in \ker \pi = V_1$. Moreover, all bases of this form span a subspace $U$ of the form we are counting, and a basis
\[
B'=\{ b_{1,1}, \dots , b_{1,d_{V_2}} , v_1'+b_{2,1}, \dots, v_{d_{U'}}'+b_{2,d_{U'}} \} 
\]
spans the same subspace as $B$ if and only if $[v_i]=[v_i']$ for $i=1,\dots,d_{U'}$, where $[v]$ denotes the image of a vector $v$ in the quotient space $V_1/V_2$.
Therefore there is a bijection between spaces $U$ of the required form and ordered sets $\{[v_1],\dots, [v_{d_{U'}}]\}$ of elements in the quotient space $V_1/V_2$.

For $i=1, \dots d_{U'}$, there are $q^{d_{V_1}-d_{V_2}}$ choices for $[v_i] \in V_1/V_2$, thus there are 
\begin{equation} \label{fixedspacecount}
q^{d_{U'}(d_{V_1}-d_{V_2})} = q^{(d_{U}-d_{V_2})(d_{V_1}-d_{V_2})}.
\end{equation}
 choices for the ordered set $\{[v_1] , \dots , [v_{d_{U'}}]\}$. 
The result follows.
\end{proof}


Given a vector space $V$ and a subspace $V_1\subseteq V$, Lemma~\ref{lemma_count_fix_int_image} can be used to count subspaces $U$ of $V$ when either $U\cap V_1$ is fixed, or the image of $U$ in $V/V_1$ is fixed, or when only the dimensions of these spaces are fixed. These results are given in the following three corollaries.

\begin{corollary} \label{lemma_lift_space}
Let $V$ be a $d_V$-dimensional vector space. Let $V_1$, $V_2$ be subspaces of $V$, of dimensions $d_{V_1}$ and $d_{V_2}$ respectively, such that $V_2 \subseteq V_1$. The number of $d_U$-dimensional subspaces $U \subseteq V$ such that $U \cap V_1 =V_2$, is given by
\[
q^{(d_{U}-d_{V_2})(d_{V_1} - d_{V_2})}
\qbinom{d_{V}-d_{V_1}}{d_{U}-d_{V_2}}.
\]
\end{corollary}

\begin{proof}
The quotient space $V / V_1$ is a space of dimension $d_V-d_{V_1}$. 
Let $U'$ be a $(d_{U}-d_{V_2})$-dimensional subspace of $V / V_1$. There are 
\begin{equation*} 
\qbinom{d_{V}-d_{V_1}}{d_{U}-d_{V_2}}
\end{equation*}
possible choices for $U'$. For each such space $U'$,  
there are $q^{(d_{U}-d_{V_2})(d_{V_1} - d_{V_2})}$ possibilities for a space $U$ with whose image in the quotient $V/V_1$ is $U'$,  by Lemma~\ref{lemma_count_fix_int_image}. 
\end{proof}


\begin{corollary} \label{thm:subspacecount}
Let $V$ be a $d_V$-dimensional vector space. Let $V_1$ be a $d_{V_1}$-dimensional subspace of $V$. The number of $d_U$-dimensional subspaces $U \subseteq V$ such that the image of $U$ in the quotient space $V/V_1$ is some fixed $d_{U'}$ dimensional space $U'$, is given by
\begin{equation} \label{decomp}
q^{d_{U'}(d_{V_1} - (d_{U}-d_{U'}))}
\qbinom{d_{V_1}}{d_U-d_{U'}}.
\end{equation}
\end{corollary}
\begin{proof}
There are
\[
\qbinom{d_{V_1}}{d_U-d_{U'}}
\]
possible choices for  a $(d_{U}-d_{U'})$-dimensional subspace $V_2$ of $V_1$. For each choice of $V_1$, there are $q^{(d_{U}-(d_U-d_{U'}))(d_{V_1} - (d_U-d_{U'}))}=q^{d_{U'}(d_{V_1} - (d_{U}-d_{U'}))}$ possibilities for the space $U$ whose intersection with $V_1$ is the fixed space $V_2$, by Lemma~\ref{lemma_count_fix_int_image}.\end{proof}

\begin{corollary} \label{cor:count_spaces_fixed_dim_int}
Let $V$ be a $d_V$-dimensional vector space. Let $V_1$ be a $d_{V_1}$-dimensional subspace of $V$. For a subspace $U$ of $V$, let $[U]$ denote the image of $U$ in the quotient space $V/V_1$. The number of $d_U$-dimensional subspaces $U \subseteq V$ such that $\dim(U \cap V_1) = d_{UV_1}$ is equal to the number of $d_U$-dimensional subspaces $U$ such that $\dim([U])=d_U-d_{UV_1}$. This number is equal to
\[
q^{(d_{U}-d_{UV_1})(d_{V_1} - d_{UV_1})}
\qbinom{d_{V}-d_{V_1}}{d_{U}-d_{UV_1}}
\qbinom{d_{V_1}}{d_{UV_1}}
.
\]
\end{corollary}

\begin{proof}
Note that $\dim(U \cap V_1)=d_{UV_1}$ if and only if $\dim([U])=d_{U}-d_{UV_1}$ hence the first statement of the lemma holds. 
Let $V_2$ be a $(d_{UV_1})$-dimensional subspace of $V_1$. There are
\[
\qbinom{d_{V_1}}{d_{UV_1}}
\]
possible choices for $V_2$, and
\begin{equation*} 
\qbinom{d_{V}-d_{V_1}}{d_{U}-d_{UV_1}}
\end{equation*}
possible choices for a $(d_{U}-d_{UV_1})$-dimensional subspace $U'$ of $V/V_1$. For each choice of $V_2$ and $U'$, Lemma~\ref{lemma_count_fix_int_image} implies there are $q^{(d_{U}-d_{UV_1})(d_{V_1} - d_{UV_1})}$ possibilities for the space $U$ whose intersection with $V_1$ is the fixed space $V_2$, and image in the quotient $V/V_1$ is the fixed space $U'$.\end{proof}


\section{Matrices over finite fields}
\label{sec:matrices}

This short section describes the notation and results we use from the theory of matrices over finite fields.

Let $q$ be a non-trivial prime power, that is $q=p^n$  for some prime $p$ and integer $n\geq1$. Let $\bF_q$ be the finite field of order $q$. In the introduction we defined $\bF_q^{n\times m}$ to be the set of $n\times m$ matrices with entries in $\bF_q$, we defined $\bF_q^{n\times m,r}$ to be the matrices in $\bF_q^{n\times m}$ of rank~$r$, and we defined $\GL(n,q)$ to be the set of invertible matrices in $\bF_q^{n\times n}$.

For a matrix $M$, we write $\rank(M)$ for the rank of $M$ and we write $\row(M)$ for the row space of $M$.

\begin{lemma}
\label{f_0_lemma}
Let $U$ be a subspace of $\bF_q^m$ of dimension $u$. The number $f_0(u)$ of matrices $M\in\bF_q^{n\times m}$ such that $\row(M)=U$ can be efficiently computed; it depends only on $q$, $n$, $m$ and $u$. For $0 \leq u \leq \min\{n,m\}$,
\begin{align} 
f_0(u) &= \prod_{i=0}^{u-1} q^n-q^i \label{eq:f0_Gabidulin}
\\
&=\sum_{v=0}^{u}  (-1)^{u-v}  q^{nv + \binom{u-v}{2}} \qbinom{u}{v} . \label{eq:f0}
\end{align}
\end{lemma}

By an efficient computation, we mean a polynomial (in $\max\{n,m\}$) number of arithmetic operations. Gabidulin~\cite[Theorem 4]{Gabidulin85} establishes~\eqref{eq:f0_Gabidulin}, and~\eqref{eq:f0} follows from~\cite[Equation 13]{Gabidulin85}. Therefore Lemma~\ref{f_0_lemma} immediately follows. 

The following results will be proved in Section~\ref{sec:matrices_proofs}.

\begin{lemma}
\label{f_1_lemma}
Let $U$ and $V$ be subspaces of $\bF_q^m$ of dimensions $u$ and $v$ respectively. Let $h=\dim (U\cap V)$. Let $M\in\bF_q^{n\times m}$ be a fixed matrix such that $\row(M)=U$. Let $r$ be a non-negative integer. The number of matrices $B\in\bF_q^{n\times m,r}$ such that $\row(B+M)=V$ can be efficiently computed; it depends only on $q$, $n$, $m$, $r$, $u$, $v$ and $h$. We write $f_1(u,v,h; r)$ for the number of matrices $B$ of this form.
\end{lemma}

\begin{lemma}
\label{f_2_lemma}
Let $r$, $r_B$ and $r_X$ be non-negative integers. Let $X$ be a fixed matrix such that $\rank(X)=r_X$. The number of matrices $B\in\bF_q^{n\times m,r_B}$ such that $\rank(X+B)=r$ can be efficiently computed; it depends only on $q$, $n$, $m$, $r$, $r_B$ and $r_X$. We write $f_2(r,r_X,r_B)$ for the number of matrices $B$ of this form.
\end{lemma}

In Section~\ref{sec:matrices_proofs}, Theorems~\ref{thm:row(M+B)=Ucount} and~\ref{f_2_thm} we give exact expressions for the functions $f_1$ and $f_2$ respectively in terms of their inputs and the values $q$, $n$ and $m$, from which Lemmas~\ref{f_1_lemma} and~\ref{f_2_lemma} follow immediately. 

We comment that the function $f_2$ has connections with rank metric codes (see e.g.~\cite{gadouleau2008packing},~\cite{silva2008rank} for example). For a fixed matrix $X$ of rank $r_X$, the function $f_2(r_X,r_B,r)$ gives the number of matrices $B$ of rank $r_B$ such that $\rank(X+B)=r$. This is equal to the number of matrices $B'$ of rank $r_B$ such that $\rank(X-B')=r$ (setting $B'=-B$).
The \emph{rank distance} is a metric defined for two matrices $M_1, M_2\in\bF_q^{n\times m}$ to be
\[
d_R(M_1,M_2)=\rank(M_1-M_2).
\]
Therefore, the value $f_2(r_X,r_B,r)$ gives the number of matrices of rank $r_B$, that have rank distance $r$ from some fixed matrix of rank $r_X$. Or equivalently, considering the space of all $n\times m$ matrices over $\bF_q$, $f_2(r_X,r_B,r)$ is the volume of intersection of two spheres with rank radii $r_X$ and $r_W$ with centres at rank distance $r$. The analysis of the volume of intersection of spheres in the rank metric space can lead to the development of covering properties for rank metric codes, as explored by Gadouleau and Yan~\cite{Gadouleau2009}. In~\cite[Lemma 1]{Gadouleau2009}, the authors give an expression for the function $f_2$, showing that indeed it is efficiently computable. The expression they give was developed using the theory of association schemes. In Section~\ref{sec:matrices_proofs} we give an expression for $f_2$ that avoids this theory, using direct counting arguments. Thus our new formula and proof give extra insight.



\section{Input and output rank distributions} \label{sec:Gamma_rank_dist}

A distribution $\mathcal{P}_{\boldsymbol{X}}$ on the input set $\mathcal{X}$ of the Gamma channel induces a distribution (the \emph{input rank distribution}) $\mathcal{R}_{\boldsymbol{X}}$ on the set of possible ranks of input matrices. Let $\mathcal{R}_{\boldsymbol{Y}}$ be the corresponding \emph{output rank distribution}, induced from the distribution on the output set of the Gamma channel. A key result (Lemma~\ref{RYdef}) is that $\mathcal{R}_{\boldsymbol{Y}}$ depends on only the channel parameters and $\mathcal{R}_{\boldsymbol{X}}$ (rather than on $\mathcal{P}_{\boldsymbol{X}}$ itself). This section aims to prove this result: it will play a vital role in the proof of Theorem~\ref{capacityUGR} below.

\begin{definition} \label{rho_def}
Let $r, r_X, r_B \in \{0, \dots , \min\{n,m\} \}$. Define
\[
\rho(r;r_X,r_B)=\frac{f_2(r,r_X,r_B)}{|\bF_q^{n \times m, r_B}|},
\]
where $f_2$ is as defined in Lemma~\ref{f_2_lemma}.
For any fixed matrix $X\in\bF_q^{n\times m, r_X}$, we see that $\rho(r;r_X,r_B)$ gives the proportion of matrices $B\in\bF_q^{n\times m,r_B}$ with $\rank(X+B)=r$. 
Let $\mathcal{R}$ be a probability distribution on the set $\{0,1,\dots,\min\{n,m\}\}$ of possible ranks of $n\times m$ matrices. Define
\[
\rho(r;r_X)
=
\sum_{r_B=0}^{\min\{n,m\}} 
\mathcal{R}(r_B) \rho(r;r_X,r_B),
\]
so that $\rho(r;r_X)$ gives the weighted average of this proportion over the possible ranks of matrices $B$.
\end{definition}

\begin{lemma} \label{lemma:rho}
Let $\boldsymbol{X}$ be an $n \times m$ matrix sampled from some distribution $\mathcal{P}_{\boldsymbol{X}}$ on~$\bF_q^{n\times m}$. Let $\boldsymbol{B}$ be an $n \times m$ matrix sampled from a UGR distribution with rank distribution $\mathcal{R}$, where $\boldsymbol{X}$ and $\boldsymbol{B}$ are chosen independently. Let $r,r_X,r_B \in \{0, \dots , \min\{n,m\} \}$. Then
\begin{equation} \label{eq:rho_lemma_1}
\rho(r;r_X,r_B)=\Prob(\rank(\boldsymbol{X}+\boldsymbol{B})=r | \rank (\boldsymbol{X}) = r_X\text{ and }
 \rank(\boldsymbol{B})=r_B),
\end{equation}
and
\begin{equation} \label{eq:rho_lemma_2}
\rho(r;r_X)=\Prob(\rank(\boldsymbol{X}+\boldsymbol{B})=r | \rank (\boldsymbol{X}) = r_X) .
\end{equation}
\end{lemma}
\begin{proof}
Let $X$ be a fixed $n \times m$ matrix of rank $r_X$. Then, since $\boldsymbol{B}$ has a UGR distribution,
\begin{align}
\Prob&(\rank(X+\boldsymbol{B})=r|\rank(\boldsymbol{B})=r_B) \notag \\
&=
\frac{
|\{B \in \mathbb{F}_q^{n \times m, r_B} : \rank(X+B)= r \} |}
{|\mathbb{F}_q^{n \times m, r_B} |} \notag
\\
&=\frac{
f_2(r,r_X,r_B)}
{|\mathbb{F}_q^{n \times m, r_B} |} \notag
\\
&=\rho(r;r_X,r_B). \label{eq:rho_lemma_pf_1}
\end{align}
Note that \eqref{eq:rho_lemma_pf_1} only depends on $\rank(X)$, not $X$ itself. Hence 
\begin{align*}
&\Prob(\rank(\boldsymbol{X}+\boldsymbol{B})=r | \rank (\boldsymbol{X}) = r_X, \rank(\boldsymbol{B})=r_B)
\\
&= \sum_X \Prob(\boldsymbol{X}=X)
\Prob(\rank(X+\boldsymbol{B})=r|\rank(\boldsymbol{B})=r_B) \notag
\\
&= \sum_{X}\Prob(\boldsymbol{X}=X)
 \rho(r;r_X,r_B) 
\\
&= \rho(r;r_X,r_B), 
\end{align*}
where the sums are over matrices $X\in \bF_q^{n \times m , r_X}$.
Thus \eqref{eq:rho_lemma_1} holds. Also
\begin{align*}
\Prob&(\rank(\boldsymbol{X}+\boldsymbol{B})=r | \rank (\boldsymbol{X}) = r_X)
\\
&=
\sum_{r_B=0}^{\min\{n,m\}} \mathcal{R}(r_B)  \rho(r;r_X,r_B) \text{ (by~\eqref{eq:rho_lemma_1})}
\\
&=\rho(r;r_X). 
\end{align*}
Thus \eqref{eq:rho_lemma_2} holds, and so the lemma follows.
\end{proof}

\begin{lemma} \label{RYdef}
For the Gamma channel $\Gamma(\mathcal{R})$ with input rank distribution $\mathcal{R}_{\boldsymbol{X}}$, the output rank distribution is given by
\[
\mathcal{R}_{\boldsymbol{Y}}(r)
=
\sum_{r_X,r_B=0}^{\min \{ n, m \} }
\mathcal{R}_{\boldsymbol{X}}(r_X) 
 \mathcal{R}(r_B) 
\frac{f_2(r,r_X,r_B)}{|\bF_q^{n \times m, r_B}|}
\]
for $r = 1, \dots, \min \{ n, m \}$. In particular, $\mathcal{R}_{\boldsymbol{Y}}$ depends only on the input rank distribution (and the channel parameters), not on the input distribution itself.
\end{lemma}
\begin{proof} We have that $\Prob(\rank(\boldsymbol{X})=r_X)=\mathcal{R}_{\boldsymbol{X}}(r_X)$ and $\Prob(\rank(\boldsymbol{B})=r_B)=\mathcal{R}(r_B)$. Hence, by~\eqref{eq:rho_lemma_1},
\begin{align*}
\mathcal{R}_{\boldsymbol{Y}}(r)&=
\Prob(\rank(\boldsymbol{Y})=r) \\
&=
\sum_{r_X,r_B=0}^{\min \{ n, m \} }
\mathcal{R}_{\boldsymbol{X}}(r_X) 
 \mathcal{R}(r_B) 
 \rho \left( r_Y ; r_X , r_B \right)
\\
&=
\sum_{r_X,r_B=0}^{\min \{ n, m \} }
\mathcal{R}_{\boldsymbol{X}}(r_X) 
 \mathcal{R}(r_B) 
\frac{f_2(r,r_X,r_B)}{|\bF_q^{n \times m, r_B}|}.\qedhere
\end{align*}
\end{proof}


\section{A UGR input distribution achieves capacity}
\label{sec:Gamma_UGR}

This section shows (Theorem~\ref{capacityUGR}) that there exists a UGR input distribution to the Gamma channel that achieves capacity.


\begin{lemma} \label{H(Y)=H(Y')}
Let $M$ and $M'$ be fixed $n \times m$ matrices of the same rank. Let $\boldsymbol{B}$ be an $n \times m$ matrix picked from a UGR distribution, and let $\boldsymbol{A}$ be an $n \times n$ matrix picked uniformly from $\GL(n,q)$, with $\boldsymbol{B}$ and $\boldsymbol{A}$ picked independently. Let $\boldsymbol{Y}=\boldsymbol{A}(M+\boldsymbol{B})$ and let $\boldsymbol{Y'}=\boldsymbol{A}(M'+\boldsymbol{B})$. Then
\[
H(\boldsymbol{Y})=H(\boldsymbol{Y'}).
\]
\end{lemma}
\begin{proof}
Let $A$ be a fixed $n \times n $ invertible matrix. 
Since the matrices $AM$ and $AM'$ have the same rank, there exist invertible matrices $R$ and $C$ such that $AM'=RAMC$. Consider the linear transformation 
$\varphi : \mathbb{F}_q^{n \times m} \rightarrow \mathbb{F}_q^{n \times m}$ defined by
$\varphi(\boldsymbol{Y}) = R\boldsymbol{Y}C$. It is simple to check that $\varphi$ is well defined and a bijection. Note that 
\begin{align*}
\varphi (A(M+\boldsymbol{B})) &= RAMC + RA\boldsymbol{B}C 
\\
 &= A(M' + A^{-1}RA\boldsymbol{B}C).
\end{align*}
Since $\boldsymbol{B}$ is picked uniformly once its rank is determined, pre- and post-multiplying $\boldsymbol{B}$ by fixed invertible matrices gives a uniform matrix of the same rank, therefore $\boldsymbol{B}$ and $A^{-1}RA\boldsymbol{B}C$ have the same distribution.
Now
\begin{align}
\Prob&\left(\boldsymbol{Y}=Y | \boldsymbol{A}=A\right)\notag\\
&= \Prob\left(A(M+\boldsymbol{B})=Y\right) \notag
\\
&= \Prob\left(\varphi(A(M+\boldsymbol{B}))=\varphi(Y)\right) \notag
\\
&= \Prob\left(A(M' + A^{-1}RA\boldsymbol{B}C)=\varphi(Y)\right) \notag
\\
 \label{Wsamedist}
&= \Prob\left(A(M' + \boldsymbol{B})=\varphi(Y)\right)
\\
 \label{P(Y'=RyC|A)=P(Y=y|A)}
&=\Prob\left(\boldsymbol{Y'}=\varphi(Y) | \boldsymbol{A}=A\right),
\end{align}
where \eqref{Wsamedist} holds since the distributions of $\boldsymbol{B}$ and $A^{-1}RA\boldsymbol{B}C$ are the same. 
Since \eqref{P(Y'=RyC|A)=P(Y=y|A)} is true for any fixed matrix $A$, it follows that
\begin{align}
\Prob(\boldsymbol{Y}=Y)
&= \!\!\!\sum_{A \in \GL(n,q)} \!\!\!\!\! \Prob(\boldsymbol{A}=A) \Prob(\boldsymbol{Y} =Y | \boldsymbol{A}=A) \notag
\\
&= \!\!\!\sum_{A \in \GL(n,q)} \!\!\!\!\! \Prob(\boldsymbol{A}=A) \Prob(\boldsymbol{Y'} =\varphi(Y) | \boldsymbol{A}=A) \notag
\\
 \label{P(Y=y)=P(Y'=RyC)}
&= \Prob (\boldsymbol{Y'}=\varphi(Y)).
\end{align}
Thus $\boldsymbol{Y}$ and $\boldsymbol{Y'}$ have the same distribution, up to relabeling by $\varphi$. In particular, we find that $H(\boldsymbol{Y})=H(\boldsymbol{Y'})$. 
\end{proof}


\begin{definition}  \label{def:hr}
Let $M$ be any $n \times m$ matrix of rank $r$.
Let $\boldsymbol{A}$ be an $n \times n$ invertible matrix chosen uniformly from $\GL(n,q)$. Let $\boldsymbol{B}$ be an $n \times m$ matrix chosen from a UGR distribution with rank distribution $\mathcal{R}$, where $\boldsymbol{A}$ and $\boldsymbol{B}$ are picked independently. 
We define
\[
h_r = H \left( \boldsymbol{A}(M+\boldsymbol{B}) \right).
\]

\end{definition}
Lemma~\ref{H(Y)=H(Y')} implies that the value $h_r$ does not depend on $M$, only on the rank $r$ and the channel parameters $q,n,m$ and $\mathcal{R}$. The exact value of $h_r$ will be calculated later, in Theorem~\ref{thm:h_r=}.


\begin{lemma} \label{H(Y|X)onlyRX}
Consider the Gamma channel $\Gamma(\mathcal{R})$. Let the input matrix $\boldsymbol{X}$ be sampled from a distribution $\mathcal{P}_{\boldsymbol{X}}$ with associated rank distribution $\mathcal{R}_{\bs{X}}$, and let $\boldsymbol{Y}$ be the corresponding output matrix. Then
\[
H(\boldsymbol{Y}|\boldsymbol{X})=  \sum_{r =0}^{\min\{n,m\}} \mathcal{R}_{\boldsymbol{X}}(r)h_r.
\]
In particular, $H(\boldsymbol{Y}|\boldsymbol{X})$ depends only on the associated input rank distribution $\mathcal{R}_{\boldsymbol{X}}$ and the channel parameters.
\end{lemma}
\begin{proof} Choosing $\boldsymbol{A}$ and $\boldsymbol{B}$ as in the definition of the Gamma channel, we see that
\begin{align*}
H(\boldsymbol{Y}|\boldsymbol{X})
&= \sum_{X \in \mathcal{X}} P(\boldsymbol{X}=X)H(\boldsymbol{A}(X+\boldsymbol{B}))
\\
&= \sum_{X \in \mathcal{X}} P(\boldsymbol{X}=X)h_{\rank(X)}
\\
&=  \sum_{r =0}^{\min\{n,m\}} \mathcal{R}_{\boldsymbol{X}}(r)h_r,
\end{align*}
which establishes the first assertion of the lemma. The second assertion follows since $h_r$ depends only on $r$ and the channel parameters.
\end{proof}


The following lemma is a well known result, see for example \cite[Ex. 2.28]{cover2012elements}.

\begin{lemma} \label{H(Y2)geqH(Y1)}
Let $\boldsymbol{Y_1}$ and $\boldsymbol{Y_2}$ be two random $n \times m$ matrices, sampled from distributions with the same associated rank distribution $\mathcal{R}_{\boldsymbol{Y}}$. If the distribution of $\boldsymbol{Y_2}$ is UGR then $H(\boldsymbol{Y_2}) \geq H(\boldsymbol{Y_1})$.
\end{lemma}

\begin{lemma} \label{YUGR}
Consider the Gamma channel $\Gamma(\mathcal{R})$. If the input distribution $\mathcal{P}_{\boldsymbol{X}}$ is UGR then the induced output distribution $\mathcal{P}_{\boldsymbol{Y}}$ is also UGR.
\end{lemma}
\begin{proof}
Suppose the input distribution is UGR, with rank distribution $\mathcal{R}_{\boldsymbol{X}}$. We start by showing that the distribution of $\boldsymbol{X}+\boldsymbol{B}$ is UGR. Let $D$ be any $n \times m$ matrix. Then
\begin{align*}
&\Prob(\boldsymbol{X}+\boldsymbol{B}= D) \notag
\\
&=\!\!\!
\sum_{X \in \mathbb{F}_q^{n \times m}} \Prob(\boldsymbol{X}=X)  \Prob(\boldsymbol{X}+\boldsymbol{B}=D | \boldsymbol{X}=X ) \notag
\\
&= \!\!\!
\sum_{X \in \mathbb{F}_q^{n \times m}}
\frac{\mathcal{R}_{\boldsymbol{X}}(\rank(X))}{|\mathbb{F}_q^{n \times m,\rank(X)}|} 
\Prob(X+\boldsymbol{B}=D) ,
\end{align*}
since $\boldsymbol{X}$ is sampled from a UGR distribution. Hence
\begin{align*}
&\Prob(\boldsymbol{X}+\boldsymbol{B}= D)\\
&= 
\sum_{r=0}^{\min\{n,m\}}
 \frac{\mathcal{R}_{\boldsymbol{X}}(r)}{|\mathbb{F}_q^{n \times m,r}|}
\sum_{X \in \mathbb{F}_q^{n \times m,r}}
\Prob(\boldsymbol{B}=D-X) \\
&=
\sum_{r=0}^{\min\{n,m\}}
 \frac{\mathcal{R}_{\boldsymbol{X}}(r)}{|\mathbb{F}_q^{n \times m,r}|}
\sum_{X \in \mathbb{F}_q^{n \times m,r}}
\frac{\mathcal{R}(\rank(D-X))}{|\mathbb{F}_q^{n \times m,\rank(D-X)}|},
\end{align*}
since $\boldsymbol{X}$ and $\boldsymbol{B}$ are independent, and since $\boldsymbol{B}$ has a UGR distribution with rank distribution $\mathcal{R}$. Now
\begin{align*}
&\sum_{X \in \mathbb{F}_q^{n \times m,r}}
\frac{\mathcal{R}(\rank(D-X))}{|\mathbb{F}_q^{n \times m,\rank(D-X)}|}\\
&=
\sum_{r_B =0}^{\min\{n,m\}}
|\{ X \in \mathbb{F}_q^{n \times m , r}: \rank(D-X)=r_B \} | 
\frac{\mathcal{R}(r_B) }{|\mathbb{F}_q^{n \times m,r_B}|}
\\
&=
\sum_{r_B =0}^{\min\{n,m\}}
f_2(r_B,\rank(D),r)  
\frac{\mathcal{R}(r_B) }{|\mathbb{F}_q^{n \times m,r_B}|} 
\end{align*}
and so
\begin{align*}
&\Prob(\boldsymbol{X}+\boldsymbol{B}= D)\\
&=\!\!\!\!\sum_{r=0}^{\min\{n,m\}}
 \frac{\mathcal{R}_{\boldsymbol{X}}(r)}{|\mathbb{F}_q^{n \times m,r}|}
\!\!\!\!\sum_{r_B =0}^{\min\{n,m\}}
f_2(r_B,\rank(D),r)  
\frac{\mathcal{R}(r_B) }{|\mathbb{F}_q^{n \times m,r_B}|}.
\end{align*}
So $\Prob(\boldsymbol{X}+\boldsymbol{B}= D)$ does not depend on the specific matrix~$D$, only its rank. Therefore, given any two $n \times m$ matrices $D_1, D_2$ of the same rank,
\[
\Prob(\boldsymbol{X}+\boldsymbol{B}=D_1) = \Prob(\boldsymbol{X}+\boldsymbol{B}=D_2).
\]
Hence $\boldsymbol{X}+\boldsymbol{B}$ has a UGR distribution. 

Let $A$ be a fixed $n \times n$ invertible matrix. 
Since $\boldsymbol{X}+\boldsymbol{B}$ is picked uniformly once its rank is determined, multiplying $\boldsymbol{X}+\boldsymbol{B}$ by the invertible matrix $A$ will give a uniform matrix of the same rank, therefore $A(\boldsymbol{X}+\boldsymbol{B})$ has a UGR distribution. So, defining $\boldsymbol{Y}=\boldsymbol{A}(\boldsymbol{X}+\boldsymbol{B})$ to be the output of the Gamma channel, we see that for any $n\times m$ matrix $Y$
\begin{align*}
\Prob(\boldsymbol{Y}=Y | \boldsymbol{A}=A)
&= \Prob(A(\boldsymbol{X}+\boldsymbol{B})=Y) 
\\
&= \frac{\Prob\left(\rank(A(\boldsymbol{X}+\boldsymbol{B}))=\rank(Y) \right)}{ |\mathbb{F}_q^{n \times m, \rank(Y)}|}
\\
&= 
\frac{\Prob\left(\rank(\boldsymbol{Y})=\rank(Y) | \boldsymbol{A}=A\right)}{ |\mathbb{F}_q^{n \times m, \rank(Y)}|},
\end{align*}
where the second equality follows since $A(\boldsymbol{X}+\boldsymbol{B})$ has a UGR distribution.
Thus 
\begin{align}
\Prob(\boldsymbol{Y}=Y)&= \!\!\!\!\sum_{A \in \GL(n,q)} \!\!\!\!\!\Prob(\boldsymbol{A}=A)  \Prob(\boldsymbol{Y}=Y | \boldsymbol{A}=A) \notag
\\
&=  \!\!\!\!\sum_{A \in \GL(n,q)} \!\!\!\!\!\Prob(\boldsymbol{A}=A)  \frac{\Prob(\rank(\boldsymbol{Y})=\rank(Y) | \boldsymbol{A}=A)}{ |\mathbb{F}_q^{n \times m, \rank(Y)}|} \notag
\\
&= \frac{1}{ |\mathbb{F}_q^{n \times m, \rank(Y)}|}  \Prob(\rank(\boldsymbol{Y})=\rank(Y) ). \label{S5eqYUGR}
\end{align}
Since \eqref{S5eqYUGR} holds for all $Y \in \mathbb{F}_q^{n \times m}$ it follows that $\boldsymbol{Y}$ has a UGR distribution.
\end{proof}


\begin{theorem} \label{capacityUGR}
For the Gamma channel $\Gamma(\mathcal{R})$, there exists a UGR input distribution that achieves channel capacity. Moreover, given any input distribution $\mathcal{P}_{\bs{X}}$ with associated rank distribution $\mathcal{R}_{\bs{X}}$, if $\mathcal{P}_{\bs{X}}$ achieves capacity then the UGR distribution with rank distribution $\mathcal{R}_{\bs{X}}$ achieves capacity.
\end{theorem}
\begin{proof}
Let $\boldsymbol{X_1}$ be a channel input, with output $\boldsymbol{Y_1}$ such that $\mathcal{P}_{\boldsymbol{X_1}}$ is a capacity achieving input distribution. That is $\max_{\mathcal{P}_{\boldsymbol{X}}} \{ I(\boldsymbol{X},\boldsymbol{Y}) \} = I(\boldsymbol{X_1},\boldsymbol{Y_1})$. Then define the input $\boldsymbol{X_2}$ with output $\boldsymbol{Y_2}$ to be distributed such that $\mathcal{P}_{\boldsymbol{X_2}}$ is the UGR distribution with $\mathcal{R}_{\boldsymbol{X_2}}=\mathcal{R}_{\boldsymbol{X_1}}$.   To prove the theorem it suffices to show $I(\boldsymbol{X_2},\boldsymbol{Y_2}) \geq I(\boldsymbol{X_1},\boldsymbol{Y_1})$.

By Lemma~\ref{RYdef}, $\mathcal{R}_{\boldsymbol{Y_2}} = \mathcal{R}_{\boldsymbol{Y_1}}$ and by Lemma~\ref{YUGR}, $\boldsymbol{Y_2}$ has a UGR distribution. Therefore, by Lemma~\ref{H(Y2)geqH(Y1)},
\begin{equation} \label{eqH(Y2)geqH(Y1)}
H(\boldsymbol{Y_2}) \geq H(\boldsymbol{Y_1}). 
\end{equation}
Also, since $\mathcal{R}_{\boldsymbol{X_2}}=\mathcal{R}_{\boldsymbol{X_1}}$, Lemma~\ref{H(Y|X)onlyRX} implies that
\begin{equation} \label{eqH(Y1|X1)=H(Y2|X2)}
H(\boldsymbol{Y_2}|\boldsymbol{X_2}) = H(\boldsymbol{Y_1}|\boldsymbol{X_1}).
\end{equation} 
Using \eqref{eqH(Y2)geqH(Y1)} and \eqref{eqH(Y1|X1)=H(Y2|X2)}, it follows that
\begin{align*}
I(\boldsymbol{X_2},\boldsymbol{Y_2}) &= H(\boldsymbol{Y_2}) - H(\boldsymbol{Y_2} | \boldsymbol{X_2})
\\
&\geq H(\boldsymbol{Y_1}) - H(\boldsymbol{Y_2} | \boldsymbol{X_2})
\\
&= H(\boldsymbol{Y_1}) - H(\boldsymbol{Y_1} | \boldsymbol{X_1})
\\
&= I(\boldsymbol{X_1},\boldsymbol{Y_1}). \qedhere
\end{align*}
\end{proof}



\section{Optimal input distributions and channel capacity} \label{sec:Gamma_capacity}

Theorem~\ref{capacityUGR} reduces the problem of computing the Gamma channel capacity to a maximisation over a set of variables of linear rather than exponential size, since the UGR distribution is determined by the distribution $\mathcal{R}_{\boldsymbol{X}}$ on a set of size $\min\{n,m\}+1$. In this section we give an expression for this maximisation problem in terms of the channel parameters and the efficiently computable functions $f_0$, $f_1$ and $f_2$ defined in Section~\ref{sec:matrices}. Since the mutual information is concave when considered as a function over possible input distributions (see e.g.~\cite[Theorem 2.7.4]{cover2012elements}), this is a concave maximisation problem and hence efficiently computable (see e.g.~\cite{boyd2004convex}).
Therefore the expression obtained provides a means for efficiently computing the exact channel capacity, and determining an optimal input rank distribution.

We begin by computing the value of $h_r$, as defined in Definition~\ref{def:hr}. This is needed to compute the maximisation problem in Corollary~\ref{capacity_corollary} that gives rise to the channel capacity.

\begin{theorem} \label{thm:h_r=}
The value $h_r$, as defined in Definition~\ref{def:hr}, is given by
\begin{multline*}
h_r
= \sum_{v=0}^{\min \{n,m\}} \sum_{h=0}^{\min\{r,v\}} 
q^{(v-h)(r-h)}
\qbinom{r}{h}
\qbinom{m-r}{v-h}
\\
\cdot
\left( \sum_{r_B=0}^{\min\{n,m\}}
\mathcal{R}(r_B) 
\frac{f_1(r,v,h;r_B)}{|\mathbb{F}_q^{n \times m, r_B}|} \right) 
 \log \left( \frac{f_0(v)}{\sum_{r_B=h}^{\min\{n,m, v+h\}}
\mathcal{R}(r_B) 
\frac{f_1(r,v,h;r_B)}{|\mathbb{F}_q^{n \times m, r_B}|} } \right). 
\end{multline*}
where $f_0$ is as defined in Lemma~\ref{f_0_lemma} and $f_1$ is as defined in Lemma~\ref{f_1_lemma}.
\end{theorem}

\begin{proof}
Let $M$ be a fixed $n \times m$ matrix of rank $r$. 
Let $\bs{Y}=\bs{A}(M+\bs{B})$,
 where $\bs{A}$ is picked uniformly from $\GL(n,q)$ and $\bs{B}$ has a UGR distribution with rank distribution $\mathcal{R}$. Then 
\[
h_r = H(\bs{A}(M+\bs{B}) | \rank(M)=r) = H(\bs{Y}).
\]
Since $\row(\bs{Y})$ is fully determined by $\bs{Y}$, it follows that $H(\bs{Y}, \row(\bs{Y}))=H(\bs{Y})$. Therefore, using the chain rule for entropy (e.g.~\cite[Thm. 2.2.1]{cover2012elements}), we have
\begin{align}
H(\bs{Y}) &= H(\bs{Y}, \row(\bs{Y})) \notag
\\
&= H(\bs{Y} | \row(\bs{Y}) ) +H(\row(\bs{Y})) \label{chainrulerowY}.
\end{align}
Now, multiplying $(M+\bs{B})$ by a uniformly picked invertible matrix will result in a uniform matrix of the same rowspace as $(M+\bs{B})$. That is, the distribution of $\bs{Y}$ is uniform given the rowspace of $\bs{Y}$. Thus (see~\cite[Thm. 2.6.4]{cover2012elements})
\begin{align}
H(\bs{Y}|\row(\bs{Y})=V)  &=  \log \left(\lvert\{Y': Y' \in \mathbb{F}_q^{n \times m}, \row(Y')=V\}\rvert \right) \notag
\\
&= \log \left( f_0(\dim(V)) \right) \label{eq:H(Y|row(Y)=U)},
\end{align} 
where $f_0$ is as defined in Lemma~\ref{f_0_lemma}.
Therefore
\begin{align}
H(\bs{Y} | \row(\bs{Y}) ) 
&= \sum_{V \subseteq \mathbb{F}_q^m} \Prob(\row(\bs{Y})=V)  H(\bs{Y}|\row(\bs{Y})=V) \notag
\\
&= \sum_{V \subseteq \mathbb{F}_q^m} \Prob(\row(\bs{Y})=V)  \log \left( f_0(\dim(V)) \right).
\end{align}
Hence
\begin{align}
h_r &= H(\bs{Y}) \notag
\\
&=H(\bs{Y} | \row(\bs{Y}) ) +H(\row(\bs{Y})) \notag
\\
&= \sum_{V \subseteq \mathbb{F}_q^m} \Prob(\row(\bs{Y})=V)  \log \left( f_0(\dim(V)) \right) \notag
\\
& \hspace{2.5cm}
- \sum_{V \subseteq \mathbb{F}_q^m} \Prob(\row(\bs{Y})=V)  \log \left( \Prob(\row(\bs{Y})=V)  \right) \notag
\\
&= \sum_{V \subseteq \mathbb{F}_q^m} \Prob(\row(\bs{Y})=V)  \log \left( \frac{f_0(\dim(V))}{\Prob(\row(\bs{Y})=V)} \right) \label{eq:hr_Prob(row(Y)=V)}
.
\end{align}
%
%
%
%
%
%
%
%
Now, we calculate the probability of $\bs{Y}$ having a particular rowspace $V$. Set $U=\row(M)$, so $\dim U=r$. For $V \subseteq \bF_q^m$, let $d_{UV} = \dim(U\cap V)$. Using the function $f_1$ defined in Lemma~\ref{f_1_lemma}, we obtain the following result. 
\begin{align}
\Prob(\row&(\bs{Y})=V) \notag
\\
&= \Prob(\row(M+\bs{B})=V) \notag
\\
&= 
\sum_{r_B=0}^{\min\{n,m\}} \Prob(\rank(\bs{B})=r_B)  \Prob(\row(M+\bs{B})=V | \rank(\bs{B})=r_B ) \notag
\\
&=
\sum_{r_B=0}^{\min\{n,m\}}
\mathcal{R}(r_B) 
\frac{|\{B: \rank(B)=r_B, \row(M+B)=V \}|}{|\mathbb{F}_q^{n \times m, r_B}|} \label{eq:probUpf1}
\\
&=
\sum_{r_B=0}^{\min\{n,m\}}
\mathcal{R}(r_B) 
\frac{f_1(r,\dim(V),d_{UV};r_B)}{|\mathbb{F}_q^{n \times m, r_B}|} 
\label{eq:probUpf3},
\end{align}
where \eqref{eq:probUpf1} follows since $\bs{B}$ has a UGR distribution.

Substituting \eqref{eq:probUpf3} into \eqref{eq:hr_Prob(row(Y)=V)} we get
\begin{multline} \label{eq:h_r_sum_over_V}
h_r
= \sum_{V \subseteq \mathbb{F}_q^m}
\left( \sum_{r_B=0}^{\min\{n,m\}}
\mathcal{R}(r_B) 
\frac{f_1(r,\dim(V),d_{UV};r_B)}{|\mathbb{F}_q^{n \times m, r_B}|} \right) 
\\
\cdot \log \left( \frac{f_0(\dim(V))}{\sum_{r_B=0}^{\min\{n,m\}}
\mathcal{R}(r_B) 
\frac{f_1(r,\dim(V),d_{UV};r_B)}{|\mathbb{F}_q^{n \times m, r_B}|} } \right)
\end{multline}
%
%
%
%
%
In \eqref{eq:h_r_sum_over_V}, for a given subspace $V \subseteq \bF_q^m$, the corresponding term in the sum depends only on $\dim(V)$ and $d_{UV}=\dim(\row(M)\cap V)$. Clearly $0\leq d_{UV}\leq\min\{\dim U, \dim V\}$. Moreover, Corollary~\ref{cor:count_spaces_fixed_dim_int} implies that the number of spaces $V$ with $\dim(V)=v$ and $\dim(\row(M)\cap V)=h$ for fixed integers $v$ and $h$ is
\[
q^{(v-h)(r-h)}
\qbinom{r}{h}
\qbinom{m-r}{v-h}.
\]
Combining this with \eqref{eq:h_r_sum_over_V} proves the theorem.
\end{proof}


Now we give the result of this section: an efficiently computable expression for the Gamma channel capacity as a maximisation over the set of possible input rank distributions. 

\begin{corollary} \label{capacity_corollary}
The capacity of the Gamma channel $\Gamma(\mathcal{R})$ is given by
\[
C=
\max_{\mathcal{R}_X} \left\{
\left(
\sum_{r_Y=0}^{\min \{ n, m \} }
\mathcal{R}_{\boldsymbol{Y}}(r_Y) 
\log
\left(
\frac{| \mathbb{F}_q^{n \times m,r_Y} |}{\mathcal{R}_{\boldsymbol{Y}}(r_Y)}  
 \right)
 \right)
-\sum_{r_X=0}^{\min\{n,m\}} 
\mathcal{R}_{\boldsymbol{X}}(r_X)h_{r_X}\right\},
\]
where $h_{r_X}$ may be computed using Theorem~\ref{thm:h_r=}, and $\mathcal{R}_{\boldsymbol{Y}}(r_Y)$ may be computed using Lemma~\ref{RYdef}.
\end{corollary}
\begin{proof}
The capacity $C$ of the channel is defined to be
\begin{equation}
\label{eqn:Cdef}
C= \max_{\mathcal{P}_{\boldsymbol{X}}} \{ I(\boldsymbol{X};\boldsymbol{Y}) \} = \max_{\mathcal{P}_{\boldsymbol{X}}} \{ H(\boldsymbol{Y}) - H(\boldsymbol{Y}|\boldsymbol{X})\}.
\end{equation}

 By Theorem~\ref{capacityUGR}, to achieve capacity we can chose the input distribution $\mathcal{P}_{\boldsymbol{X}}$ to be UGR. By Lemma~\ref{YUGR}, the output distribution will also be UGR. Therefore the output distribution is given by
\begin{equation}
\mathcal{P}_{\boldsymbol{Y}}(Y) = \Prob(\boldsymbol{Y}=Y) = \frac{1}{| \mathbb{F}_q^{n \times m,\rank(Y)} |}  \mathcal{R}_{\boldsymbol{Y}}(\rank(Y))  \label{Py}
\end{equation}
for any $Y \in \mathbb{F}_q^{n \times m}$. Thus the entropy of $\boldsymbol{Y}$ is given by
\begin{align*}
H(\boldsymbol{Y}) 
&= - \sum_{Y \in \mathbb{F}_q^{n \times m}} \Prob(\boldsymbol{Y}=Y)  \log \Prob(\boldsymbol{Y}=Y) \notag
\\
&=  - \sum_{Y \in \mathbb{F}_q^{n \times m}} 
\left(
\frac{1}{| \mathbb{F}_q^{n \times m,\rank(Y)} |}  \mathcal{R}_{\boldsymbol{Y}}(\rank(Y)) 
\right )
\log
\left(
\frac{1}{| \mathbb{F}_q^{n \times m,\rank(Y)} |}  \mathcal{R}_{\boldsymbol{Y}}(\rank(Y)) 
\right) \notag
\\
&= - \sum_{r_Y =0}^{\min\{n,m\}}
\mathcal{R}_{\boldsymbol{Y}}(r_Y) 
\log
\left(
\frac{1}{| \mathbb{F}_q^{n \times m,r_Y} |}  \mathcal{R}_{\boldsymbol{Y}}(r_Y) 
\right) .
\end{align*}
Since $H(\boldsymbol{Y}|\boldsymbol{X})= \sum_{r_X=0}^{\min\{n,m\}} 
\mathcal{R}_{\boldsymbol{X}}(r_X)h_{r_X}$ by Lemma~\ref{H(Y|X)onlyRX},
the result follows from~\eqref{eqn:Cdef}. 
\end{proof}



\section{Matrix function proofs} \label{sec:matrices_proofs}
The aim of this section is to derive efficiently computable expressions for the functions $f_1$ and $f_2$, thus proving Lemmas~\ref{f_1_lemma} and~\ref{f_2_lemma} respectively and providing a method for computing the capacity formula given in Corollary~\ref{capacity_corollary}.

We approach this problem by first exploring several combinatorial results. In Subsection~\ref{sec:mobius} we establish a counting result we need later, using M\"obius theory. In Subsection~\ref{sec:MF_counting_subspaces}, we use this result to derive expressions for the functions $f_1$ and $f_2$.

\subsection{A counting lemma} \label{sec:mobius}

In this subsection, we prove an `inversion' lemma, Lemma~\ref{lemma:mob_inv_prod}, that we need in the following subsection. We use M\"obius theory (a generalisation of inclusion--exclusion) to establish this lemma: see Bender and Goldman~\cite{Bender1975}, for example, for a nice introduction to this theory and an exposition of all the results we use here.

Let $\text{Po}(\bF_q^m)$ denote the poset of all subspaces of $\bF_q^m$ ordered by containment.
Let $P$ and $Q$ be two posets. Recall that the direct product $P \times Q$ is the poset where $(p_1,q_1) \leq (p_2,q_2)$ if and only if $p_1 \leq p_2$ and $q_1 \leq q_2$, where $p_1, p_2 \in P$ and $q_1, q_2 \in Q$.  

\begin{lemma} \label{lemma:mob_inv_prod}
Let $f((U,V))$ be a real valued function defined for all pairs $(U,V) \in \mathrm{Po}(\bF_q^m) \times \mathrm{Po}(\bF_q^m)$. If
\[
g((U,V)) = \sum_{(U',V') \leq (U,V)} f((U',V'))
\]
then
\[
f((U,V)) =  \sum_{(U',V') \leq (U,V)}  (-1)^{u-u'+v-v'}  q^{ \binom{u-u'}{2}+ \binom{v-v'}{2}}  g((U',V')),
\]
where $\dim(U)=u, \dim(U')=u', \dim(V)=v$ and $\dim(V')=v'$.
\end{lemma}

\begin{proof}
By the M\"obius inversion formula (see~\cite[Theorem~1]{Bender1975}, for example) 
\begin{equation} \label{pf:mob_inv_prod1}
f((U,V)) =  \sum_{(U',V') \leq (U,V)} \mu ((U',V'),(U,V))  g((U',V')),
\end{equation}
where $\mu$ is the M\"obius function of $\text{Po}(\bF_q^m) \times \text{Po}(\bF_q^m)$. But (see~\cite[Theorem~3]{Bender1975}, for example),
\[
\mu ((U',V'),(U,V))=\mu'(U',U)\mu'(V',V)
\]
where $\mu'$ is the M\"obius function of $\text{Po}(\bF_q^m)$. Moreover (see, for example, \cite[\S 5]{Bender1975}) the M\"obius function of $\text{Po}(\bF_q^m)$ may be written explicitly as
\[
\mu'(X,Y) = (-1)^{\dim(Y)-\dim(X)}  q^{\binom{\dim(Y)-\dim(X)}{2}}
\]
for any $X,Y\in \text{Po}(\bF_q^m)$ with $X\subseteq Y$.
So the lemma follows.
\end{proof}

\subsection{Computing $f_1$ and $f_2$} \label{sec:MF_counting_subspaces}

By `basic dimension properties', we mean that all specified dimensions are non-negative integers, and if dimensions $d_U$ and $d_V$ of subspaces $U\subseteq V$ are specified, then $d_U\leq d_V$.

\begin{lemma}
\label{lemma:direct_sum}
Let $z$ be a non-negative integer. Let $X$ and $Y$ be subspaces of $\bF_q^z$ of dimensions $d_X$ and $d_Y$ respectively such that $X\cap Y=\{0\}$. Let $c(d_X,d_Y,z,d_R,d_{RX},d_{RY})$ be the number of $d_R$-dimensional subspaces $R$ of $\bF_q^z$ such that $\dim(R\cap X)=d_{RX}$, $\dim(R\cap Y)=d_{RY}$ and such that $X\subseteq Y+R$. If the basic dimension properties are satisfied then $c(d_X,d_Y,z,d_R,d_{RX},d_{RY})$ is given by the formula
\[
\qbinom{d_X}{d_{RX}}\qbinom{d_Y}{d_{RY}}\qbinom{z-d_X-d_Y}{d_R-d_{RY}-d_X}q^{(d_Y-d_{RY})(d_R-d_X-d_{RY})}\prod_{i=0}^{d_X-d_{RX}-1}(q^{d_Y-d_{RY}}-q^i),
\]
otherwise $c(d_X,d_Y,z,d_R,d_{RX},d_{RY})=0$.
\end{lemma}
We remark that, in this case, the basic dimension properties are that $0\leq d_{RX}\leq \min\{d_R,d_X\}$, $0\leq d_{RY}\leq \min\{d_R,d_Y\}$ and $d_X+d_Y\leq z$. 
\begin{proof}
Suppose that $d_X> d_R-d_{RY}$.
The condition that $X\subseteq Y+R$ is equivalent to the condition that the subspace $(X+Y)/Y$ is contained in the subspace $(R+Y)/Y$. The dimensions of these subspaces are $d_X$ and $d_R-d_{RY}$ respectively, and so our count is zero in this case. But $\qbinom{z-d_X-d_Y}{d_R-d_{RY}-d_X}=0$ when $d_R-d_{RY}-d_X<0$ and so the lemma follows in this case. So we may assume that $d_X\leq d_R-d_{RY}$.

There are $\qbinom{d_X}{d_{RX}}$ choices for the subspace $R\cap X$ and $\qbinom{d_Y}{d_{RY}}$ choices for the subspace $R\cap Y$. Assume that these subspaces are fixed. The quotient $Q=(R+Y)/Y$ of $R$ in $\bF_{q}^z/Y$ has dimension $(d_R-d_{RY})$. The condition that $X\subseteq Y+R$ implies that $(X+Y)/Y\subseteq Q$. Since $(X+Y)/Y$ has dimension $d_X$, the number of choices for $Q$ is therefore $\qbinom{z-d_Y-d_X}{d_R-d_{RY}-d_X}$. Assume that $Q$ is now also fixed.

Fix $u_1,u_2,\ldots,u_{d_R-d_{RY}-d_X}\in\bF_q^m$  with the property that $\{u_i+Y:1\leq i\leq d_R-d_{RY}-d_X \}$ spans a complement to $(X+Y)/Y$ in $Q$. Fix a basis $x_1,\dots, x_{d_{RX}}$ of $R\cap X$, and extend this basis to a basis $x_1,x_2,\ldots,x_{d_X}$ of $X$. Fix a basis $y_1,y_2,\ldots,y_{d_{RY}}$ of $R\cap Y$. Every subspace $R$ we are counting has a basis of the form
\begin{multline*}
\{y_i:1\leq i\leq d_{RY}\}\cup \{x_i:1\leq i\leq d_{RX}\}\cup \{x_i+\epsilon_i:d_{RX}+1\leq i\leq d_{X}\}\\
\cup\{u_i+\delta_i:1\leq i\leq d_{R}-d_{RY}-d_X\}
\end{multline*}
for some $\epsilon_i,\delta_i\in Y$. Note that all subspaces with a basis of this form intersect $Y$ in precisely the space spanned by $\{y_i:1\leq i\leq d_{RY}\}$, and all subspaces are equal to $Q$ after taking a quotient by $Y$. Moreover, a subspace of this form intersects $X$ in precisely the subspace spanned by $\{x_i:1\leq i\leq d_{RX}\}$ if and only if the vectors $\epsilon_i+(R\cap Y)$ are linearly independent in $Y/(R\cap Y)$. Finally, two subspaces of this form are distinct if and only if the ordered set of vectors $\epsilon_i$ and $\delta_i$ are different modulo $R\cap Y$. There are $q^{(d_Y-d_{RY})(d_R-d_X-d_{RY})}$ choices for vectors $\delta_i+(R\cap Y)\in Y/(R\cap Y)$, and there are $\prod_{i=0}^{d_X-d_{RX}-1}(q^{d_Y-d_{RY}}-q^i)$ choices for linearly independent vectors $\epsilon_i+(R\cap Y)\in Y/(R\cap Y)$. So the lemma follows.
\end{proof}

We define a function $f_1(d_U,d_V,d_{UV}; r)$ as follows. When $r<(d_U-d_{UV})$, we define $f_1(d_U,d_V,d_{UV}; r)=0$. Otherwise we proceed as follows. For integers $d_{V'}$, $d_{W'}$ and $d_{V'W'}$, define
\[
\kappa_1(d_{V'},d_{W'},d_{V'W'})=(-1)^{(r-d_{W'})} q^{\binom{r-d_{W'}}{2}}(-1)^{(d_V-d_{V'})} q^{\binom{d_V-d_{V'}}{2}}q^{n d_{V'W'}}.
\]
For integers $d_{V'}$, $d_{W'}$, $d_{V'W'}$, $d_{UW'}$, $d_{VW'}$ and $d_{UVW'}$, define
\[
\kappa_2(d_{V'}, d_{W'}, d_{V'W'}, d_{UW'}, d_{VW'},d_{UVW'})=\nu_1\nu_2\nu_3
\]
where
\begin{align*}
\nu_1&=c(d_U-d_{UV},d_V-d_{UV},m-d_{UV},d_{W'}-d_{UVW'},d_{UW'}-d_{UVW'},d_{VW'}-d_{UVW'})\\
\nu_2&=\qbinom{d_{UV}}{d_{UVW'}}\qbinom{d_{VW'}}{d_{V'W'}}\qbinom{(d_V-d_{VW'})-(d_U-d_{UW'})}{(d_{V'}-d_{V'W'})-(d_U-d_{UW'})}\qbinom{m-d_{W'}}{r-d_{W'}}\text{ and}\\
\nu_3&=q^{(d_{W'}-d_{UVW'})(d_{UV}-d_{UVW'})}q^{(d_{VW'}-d_{V'W'})(d_{V'}-d_{V'W'})},
\end{align*}
where $c$ is the function defined in Lemma~\ref{lemma:direct_sum}. Then $f_1(d_U,d_V,d_{UV}; r)$ is equal to
\begin{equation*} \label{eq:f1}
\sum_{d_{W'}=0}^r\sum_{d_{V'}=0}^{d_V}\sum_{d_{V'W'}=0}^{\min\{d_{V'},d_{W'}\}}\sum_{d_{UW'}=0}^{\min\{d_U,d_{W'}\}}
\sum_{d_{VW'}=0}^{\min\{d_V,d_{W'}\}}
\sum_{d_{UVW'}=0}^{\min\{d_{UV},d_{VW'},d_{UW'}\}}\kappa_1\kappa_2,
\end{equation*}
where $\kappa_1=\kappa_1(d_{V'},d_{W'},d_{V'W'})$ and $\kappa_2=\kappa_2(d_{V'}, d_{W'}, d_{V'W'}, d_{UW'}, d_{VW'},d_{UVW'})$.

\begin{theorem} \label{thm:row(M+B)=Ucount}
Let $f_1$ be as defined in Lemma~\ref{f_1_lemma}. That is, if $U$ and $V$ are subspaces of $\bF_q^m$ of dimensions $d_U$ and $d_V$ respectively, with $d_{UV}= \dim(U\cap V)$ and $M\in\bF_q^{n\times m}$ is a fixed matrix such that $\row(M)=U$; then $f_1(d_U,d_V,d_{UV}; r)$ gives the number of matrices $B\in\bF_q^{n\times m,r}$ such that $\row(M+B)=V$. Then the value $f_1(d_U,d_V,d_{UV}; r)$ is as given above. 
\end{theorem}
\begin{proof}
We begin the proof with a simpler counting problem, and then use this result to establish the formula we are aiming for.

For a subspace $W$ of $\bF_q^m$, let $g(V,W)$ be the number of $n\times m$ matrices $B$ with $\row(B)\subseteq W$ and $\row(M+B)\subseteq V$. We claim that
\[
g(V,W)=\begin{cases}
q^{n\,d_{VW}}&\text{ if }U\subseteq V+W\\
0&\text{ otherwise.}
\end{cases}
\]
To see this, we proceed as follows. Let $x_1,x_2,\ldots,x_n\in\bF_q^m$ be the rows of $M$. Suppose that $U\not\subseteq V+W$. Then $(x_i+W)\cap V=\emptyset$ for some $i$, and so we must have $g(V,W)=0$, since there is no valid choice for the $i$th row of $B$ in this case. Now suppose that  $U\subseteq V+W$, so for all $i$ we have $(x_i+W)\cap V\not=\emptyset$ and therefore there exist $w_1,w_2,\ldots ,w_n\in W$ such that $x_i+w_i\in V$. It is not hard to check that a matrix $B$ with rows $b_i$ has the property that $\row(B)\subseteq W$ and $\row(M+B)\subseteq V$ if and only if $b_i-w_i\in V\cap W$. Hence there are $q^{d_{VW}}$ choices for each row $b_i$ of $B$. Since $B$ has $n$ rows, the claim follows.

Let $f(V,W)$ be the number of $n\times m$ matrices $B$ with $\row(B)= W$ and $\row(M+B)=V$. Now $g(V,W)=\sum_{(V',W')} f(V,W)$, where the sum is over all pairs of subspaces $(V',W')$ with $V'\subseteq V$ and $W'\subseteq W$. So, by Lemma~\ref{lemma:mob_inv_prod},
\begin{align*}
f(V,W)&=\sum_{(V',W')} (-1)^{(d_W-d_{W'})+(d_V-d_{V'})}q^{\binom{d_W-d_{W'}}{2}+\binom{d_V-d_{V'}}{2}}g(V',W')\\
&=\sum_{\substack{(V',W')\\U\subseteq V'+W'}} (-1)^{(d_W-d_{W'})+(d_V-d_{V'})}q^{\binom{d_W-d_{W'}}{2}+\binom{d_V-d_{V'}}{2}}q^{nd_{V'W'}}\\
&=\sum_{\substack{(V',W')\\U\subseteq V'+W'}} \kappa_1(d_{V'},d_{W'},d_{V'W'}),
\end{align*}
where again $V'\subseteq V$ and $W'\subseteq W$ in our sums.

The number of matrices $B$ of rank $r$ such that $\row(M+B)=V$ is
\[
\sum_{\substack{W\subseteq \bF_q^m\\ \dim W=r}}f(V,W).
\]
So we can express this count as
\begin{equation}
\label{eqn:sum_expression}
\sum_{d_{W'}=0}^r\sum_{d_{V'}=0}^{d_V}\sum_{d_{V'W'}=0}^{\min\{d_{V'},d_{W'}\}}\sum_{d_{UW'}=0}^{\min\{d_U,d_{W'}\}}
\sum_{d_{VW'}=0}^{\min\{d_V,d_{W'}\}}
\sum_{d_{UVW'}=0}^{\substack{\min\{d_{UV},d_{VW'},\\d_{UW'}\}}}\sum_{V',W',W}\kappa_1(d_{V'},d_{W'},d_{V'W'}),
\end{equation}
where the last sum is over all triples $(V',W',W)$ of subspaces of $\bF_q^m$ with $V'\subseteq V$, $W'\subseteq W$, $U\subseteq V'+W'$, $\dim(W')=d_{W'}$, $\dim(W)=r$, $\dim(V')=d_{V'}$, $\dim(V'\cap W')=d_{V'W'}$, $\dim(U\cap W')=d_{UW'}$, $\dim(V\cap W')=d_{VW'}$ and $\dim{U\cap V\cap W'}=d_{UVW'}$.

We aim to count the number of possibilities for a subspace $W'$ such that $\dim(W')=d_{W'}$, $\dim(U\cap W')=d_{UW'}$, $\dim(V\cap W')=d_{VW'}$ and $\dim{U\cap V\cap W'}=d_{UVW'}$ and that satisfy the weaker condition that $U\subseteq V+W'$. We will show (see below) that the number of such subspaces $W'$ is
\begin{equation}
\label{eqn:w_dash_count}
c\qbinom{d_{UV}}{d_{UVW'}}q^{(d_{W'}-d_{UVW'})(d_{UV}-d_{UVW'})},
\end{equation}
where $c=c(d_U-d_{UV},d_V-d_{UV},m-d_{UV},d_{W'}-d_{UVW'},d_{UW'}-d_{UVW'},d_{VW'}-d_{UVW'})$ is defined in Lemma~\ref{lemma:direct_sum}.

Once we have fixed such a subspace $W'$, we choose $V'$ and $W$ as follows. We first choose the subspace $V'\cap W'$. There are $\qbinom{d_{VW'}}{d_{V'W'}}$ choices for this subspace. The quotient space $(V'+W')/W'$ of $V'$ by $W'$ is a space of dimension $d_{V'}-d_{V'W'}$. It is contained in the $(d_V-d_{VW'})$-dimensional space  $(V+W')/W'$ and contains the $(d_{U}-d_{UW'})$-dimensional space $(U+W')/W'$. So the number of choices for $(V'+W')/W'$ is
\[
\qbinom{(d_V-d_{VW'})-(d_U-d_{UW'})}{(d_{V'}-d_{V'W'})-(d_U-d_{UW'})}.
\]
Once this quotient space is also fixed, there are $q^{(d_{VW'}-d_{V'W'})(d_{V'}-d_{V'W'})}$ choices for $V'$.
Finally we choose the $r$-dimensional subspace $W$ containing $W'$: there are $\qbinom{m-d_{W'}}{r-d_{W'}}$ choices for $W$.

Combining the formula~\eqref{eqn:w_dash_count} with \eqref{eqn:sum_expression} and the counting argument of the previous paragraph, the theorem follows. So it remains to establish~\eqref{eqn:w_dash_count}.

The number of choices~\eqref{eqn:w_dash_count} for $W'$ may be found as follows. There are $\qbinom{d_{UV}}{d_{UVW'}}$ choices for the subspace $T=(U\cap V)\cap W'$. Suppose that $T$ is now fixed. We now consider the images $X$, $Y$ and $R$ of $U$, $V$ and $W'$ respectively in the quotient by $U\cap V$. So $X=(U+(U\cap V))/(U\cap V)$ has dimension $d_U-d_{UV}$ and $Y=(V+(U\cap V))/(U\cap V)$ has dimension $d_V-d_{UV}$. Moreover $R$ is a subspace of dimension $d_{W'}-d_{UVW'}$ which intersects $X$ and $Y$ in subspaces of dimension $d_{UW'}-d_{UVW'}$ and $d_{VW'}-d_{UVW'}$ respectively. The subspaces $X$ and $Y$ intersect trivially. Since $U\subseteq V+W'$, we see that $X\subseteq Y+R$. Hence, by Lemma~\ref{lemma:direct_sum}, the number of choices for the subspace $R$ is $c(d_U-d_{UV},d_V-d_{UV},m-d_{UV},d_{W'}-d_{UVW'},d_{UW'}-d_{UVW'},d_{VW'}-d_{UVW'})$. Suppose now that $R$ is fixed. There are $q^{(d_{W'}-d_{UVW'})d_{UVW'}}$ subspaces $W'$ with $(W'+(U\cap V))/(U\cap V)=R$ and $(U\cap V)\cap W'=T$. Since all of these subspace have the property that $U\subseteq W'+V$, the formula~\eqref{eqn:w_dash_count} follows, and so the theorem is proved. 
\end{proof}

\begin{theorem}
\label{f_2_thm}
Let $f_2$ be as defined in Lemma~\ref{f_2_lemma}. That is, for a fixed matrix $X$ with $\rank(X)=r_X$, $f_2(r,r_X,r_B)$ gives the number of matrices $B\in\bF_q^{n\times m,r_B}$ such that $\rank(X+B)=r$. Then,
\[
f_2(r,r_X,r_B)=
\sum_{h=0}^{\min\{ r, r_X\}} q^{(r-h)(r_X-h)}  \qbinom{m-r_X}{r-h} 
\qbinom{r_X}{h}
f_1(r_X,r,h;r_B).
\]
\end{theorem}

\begin{proof}
Using the definition of $f_1$ given above Theorem~\ref{thm:row(M+B)=Ucount}, we see that
\begin{align}
&f_2(r,r_X,r_B)
\notag
\\
&=
\sum_{V \subseteq \bF_q^m : \dim(V)=r} f_1(r_X,r,\dim(V \cap \row(X)) ; r_B)
\label{eq:f_2_proof_1}
\\
&=
\sum_{h=0}^{\min\{ r, r_X \}}
| \{ V \subseteq \bF_q^m : \dim(V)=r, \dim(V \cap \row(X))=h \} | 
  f_1(r_X,r,h ; r_B)
\label{eq:f_2_proof_2}
\end{align}
where \eqref{eq:f_2_proof_1} follows since the number of matrices $B$ with $\rank (X+B)=r$ is equal to the number of matrices $B$ with $\row(X+B)=V$, summed over all spaces $V \subseteq \bF_q^m$ with $\dim(V)=r$. 

By Corollary~\ref{cor:count_spaces_fixed_dim_int}, the number of $r$-dimensional subspaces $V \subseteq \bF_q^m$, with $\dim(V \cap \row(X))=h$ is
\begin{equation} \label{eq:f_2_proof_3}
q^{(r-h)(r_X-h)}  \qbinom{ m-r_X}{ r-h} 
\qbinom{ r_X }{ h}.
\end{equation}
Substituting \eqref{eq:f_2_proof_3} into \eqref{eq:f_2_proof_2} gives the result. 
\end{proof}

\section{Conclusion}
\label{sec:conclusion}
 
In this paper we have considered a class of matrix channels (Gamma channels) suitable for modelling random linear network coding when random errors are introduced during transmission. The Gamma channels are a generalisation of the AMMC channel considered in~\cite{silva2010}. Random errors are modelled by a matrix whose rank represents the number of linearly independent errors. The error matrix is chosen by first picking its rank according to a rank distribution $\mathcal{R}$ dependent on the application, and then choosing uniformly from all matrices of this rank (a UGR distribution). We show that in this model there always exists a capacity achieving input distribution that is UGR. This key result allows us to compute the capacity of the channel as a maximisation problem over possible (input) rank distributions, a set of linear rather than exponential size. We presented sample capacity computations in the introduction: all computations used a simple hill-climbing algorithm to perform the maximisation, and were implemented in Mathematica~10.4~\cite{Wolfram}.

\begin{problem}
Can bounds for the AMMC capacity be improved, to give good asymptotic results in more situations?
\end{problem}
We ran simulations to show that for the AMMC channel with two errors, the true capacity of the channel closely follows the trend of the previously known upper bound for the capacity. It might be possible to improve the lower bound on the capacity by using simulation results as a guide.

\begin{problem}
Can good asymptotic bounds on the capacity of the Gamma channel be established?
\end{problem}
We believe it will be hard to find good capacity bounds that hold in complete generality. But it would be very interesting to investigate the binomial rank distribution for errors, or the distribution arising for errors that are not linearly independent mentioned in the introduction.

\begin{problem}
Can explicit good coding schemes for the Gamma channel be constructed?
\end{problem}
Theorem~\ref{capacityUGR} shows that there are UGR input distributions that achieve capacity. It would be interesting to see explicit good coding schemes that use UGR input distributions. (We are not aware of such schemes, even in special cases such as the AMMC channel.)  


\section*{Acknowledgements} Jessica Claridge would like to acknowledge the support of an EPSRC PhD studentship. Both authors would like to acknowledge the support of the EU COST Action IC1104, and would like thanks reviewers of an earlier version of this paper for their comments.



\begin{thebibliography}{10}

\bibitem{ahlswede2000}
R.~Ahlswede, N.~Cai, S.-Y.~R. Li, and R.~W. Yeung.
\newblock Network information flow.
\newblock {\em IEEE Transactions on Information Theory}, 46(4):1204--1216, Jul
  2000.

\bibitem{Bender1975}
E.~A. Bender and J.~R. Goldman.
\newblock On the applications of {M}{\"o}bius inversion in combinatorial
  analysis.
\newblock {\em The American Mathematical Monthly}, 82(8):789--803, October
  1975.

\bibitem{boyd2004convex}
S.~Boyd and L.~Vandenberghe.
\newblock {\em Convex Optimization}.
\newblock Cambridge university press, 2004.

\bibitem{Cameron1994}
P.~J. Cameron.
\newblock {\em Combinatorics: {T}opics, Techniques, Algorithms}.
\newblock Cambridge University Press, 1994.

\bibitem{Claridge_thesis}
J.~Claridge.
\newblock {\em On Matrix Models for Network Coding}.
\newblock PhD Thesis, Royal Holloway, University of London, 2017.

\bibitem{cover2012elements}
T.~M. Cover and J.~A. Thomas.
\newblock {\em Elements of information theory}.
\newblock John Wiley \& Sons, 2012.

\bibitem{Gabidulin85}
{\`{E}}.~M. Gabidulin.
\newblock Theory of codes with maximum rank distance.
\newblock {\em Problems of Information Transmission}, 21(1):1--12, January
  1985.

\bibitem{gadouleau2008packing}
M.~Gadouleau and Z.~Yan.
\newblock Packing and covering properties of rank metric codes.
\newblock {\em IEEE Transactions on Information Theory}, 54(9):3873--3883,
  2008.

\bibitem{Gadouleau2009}
M.~Gadouleau and Z.~Yan.
\newblock Bounds on covering codes with the rank metric.
\newblock {\em IEEE Communications Letters}, 13(9):691--693, 2009.

\bibitem{ho2006RLNC}
T.~Ho, M.~M\'{e}dard, R.~K\"{o}tter, D.~R. Karger, M.~Effros, J.~Shi, and
  B.~Leong.
\newblock A random linear network coding approach to multicast.
\newblock {\em IEEE Transactions on Information Theory}, 52(10):4413--4430,
  October 2006.

\bibitem{Wolfram}
Wolfram~Research{,} Inc.
\newblock Mathematica, {V}ersion 10.4.
\newblock Champaign, IL, 2016.

\bibitem{Koetter2008}
R.~K\"{o}tter and F.~R. Kschischang.
\newblock Coding for errors and erasures in random network coding.
\newblock {\em IEEE Transactions on Information Theory}, 54(8):3579--3591, Aug
  2008.

\bibitem{li2003linear}
S.-Y.~R. Li, R.~W. Yeung, and N.~Cai.
\newblock Linear network coding.
\newblock {\em IEEE Transactions on Information Theory}, 49(2):371--381, 2003.

\bibitem{montanari2013}
A.~Montanari and R.~L. Urbanke.
\newblock Iterative coding for network coding.
\newblock {\em IEEE Transactions on Information Theory}, 59(3):1563--1572,
  2013.

\bibitem{Nobrega2013}
R.~W. Nobrega, D.~Silva, and B.~F. Uchoa-Filho.
\newblock On the capacity of multiplicative finite-field matrix channels.
\newblock {\em IEEE Transactions on Information Theory}, 59(8):4949--4960,
  2013.

\bibitem{siavoshani2011}
M.~J. Siavoshani, S.~Mohajer, C.~Fragouli, and S.~N. Diggavi.
\newblock On the capacity of noncoherent network coding.
\newblock {\em IEEE Transactions on Information Theory}, 57(2):1046--1066,
  2011.

\bibitem{silva2008rank}
D.~Silva, F.~R. Kschischang, and R.~K\"{o}tter.
\newblock A rank-metric approach to error control in random network coding.
\newblock {\em IEEE Transactions on Information Theory}, 54(9):3951--3967,
  2008.

\bibitem{silva2010}
D.~Silva, F.~R. Kschischang, and R.~K\"{o}tter.
\newblock Communication over finite-field matrix channels.
\newblock {\em IEEE Transactions on Information Theory}, 56(3):1296--1305,
  March 2010.

\end{thebibliography}
\end{document}